\documentclass[letterpaper,USenglish,11pt]{article}

\usepackage{comment}
\usepackage{ifthen}
\usepackage{xspace}
\usepackage{float}
\usepackage{subcaption}
\usepackage{amsmath}
\usepackage{amsthm}
\usepackage{amstext}
\usepackage{amssymb}
\usepackage{amsfonts}
\usepackage{bm}
\usepackage[multiple]{footmisc}
\usepackage{wrapfig}
\usepackage{url}
\usepackage[margin=1in]{geometry}
\usepackage{mathdots}
\usepackage{mathtools}
\usepackage{thmtools}
\usepackage{thm-restate}
\usepackage{tikz}
\usetikzlibrary{calc, fit, shapes,intersections,arrows,matrix, topaths}
\usetikzlibrary{decorations,decorations.pathreplacing,patterns, shadows}
\usepackage{graphicx}
\usepackage[linesnumbered,noresetcount,vlined,ruled]{algorithm2e}
\usepackage[inline]{enumitem}
\usepackage{environ}
\usepackage{dsfont}
\usepackage{cite}
\usepackage{yfonts} 

\newcommand{\sbs}{\text{$b^2$-Bucket Sort}\xspace}
\newcommand{\mbs}{\text{$b\log b$-Bucket Sort}\xspace}
\newcommand{\calA}{{\cal A}}

\newcommand{\calS}{{\cal S}}
\newcommand{\bB}{\mathbf{b}}
\newcommand{\bX}{X}
\newcommand{\bL}{L}
\newcommand{\bC}{\mathbf{c}}

\renewcommand{\epsilon}{\varepsilon} 

\renewcommand{\epsilon}{\varepsilon}
\newcommand{\eps}{\varepsilon}

\usepackage{amsmath} \usepackage{amsthm}
\newtheorem{theorem}{Theorem}
\newtheorem{lemma}[theorem]{Lemma}
\newtheorem{corollary}[theorem]{Corollary}

\newtheorem{definition}[theorem]{Definition}

\newtheorem{observation}[theorem]{Observation}
\newtheorem{claim2}[theorem]{Claim}

\newcommand{\size}[1]{\ensuremath{\left|#1\right|}}
\newcommand{\set}[1]{\left\{ #1 \right\}}
\newcommand{\parentheses}[1]{\left(#1\right)}
\DeclarePairedDelimiter{\floor}{\lfloor}{\rfloor}
\DeclarePairedDelimiter{\ceil}{\lceil}{\rceil}
\newcommand{\expectation}[2]{\mathbb{E}_{#1}\left[ #2 \right]}
\renewcommand{\Pr}[1]{{\mathrm{Pr}}\left[ #1 \right]}
\newcommand{\Prr}[2]{\mathrm{Pr}_{#1}\left[ #2 \right]}

\begin{document}
\title{Upper Tail Analysis of Bucket Sort and Random Tries\thanks{This research was supported by a
grant from the United States-Israel Binational Science Foundation
(BSF), Jerusalem, Israel, and the United States National Science
Foundation (NSF)}}
\author{
 Ioana O. Bercea\thanks{
Tel Aviv University, Tel Aviv, Israel.
Email:~\texttt{ioana@cs.umd.edu, guy@eng.tau.ac.il}.}
\and
Guy Even\footnotemark[1]
}
\date{}
\maketitle

\begin{abstract}
Bucket Sort is known to run in expected linear time when the input keys
are distributed independently and uniformly at random in the interval $[0,1)$. The analysis holds even when
a quadratic time algorithm is used to sort the keys in each bucket. We show how to obtain linear time guarantees on the running time of Bucket Sort that hold with
\textit{very high probability}. Specifically, we investigate the asymptotic behavior of the exponent in the upper tail probability of the running time of Bucket Sort. We consider large additive deviations from the expectation, of the form
$cn$ for large enough (constant) $c$, where $n$ is the number of keys that are sorted.

Our analysis
shows a profound difference between variants of Bucket Sort that use a quadratic time
algorithm within each bucket and variants that use a $\Theta(b\log b)$ time algorithm for sorting
$b$ keys in a bucket. When a quadratic time algorithm is used to sort the keys in a bucket, the probability that Bucket Sort takes $cn$ more time than expected is exponential in $\Theta(\sqrt{n}\log n)$. When a $\Theta(b\log b)$ algorithm is used to sort the keys in a bucket, the exponent becomes $\Theta(n)$. We prove this latter theorem by showing an upper bound on the tail of a  random variable defined on tries, a result which we believe is of independent interest. This result also enables us to analyze the upper tail probability of a well-studied trie parameter, the external path length, and show that the probability that it deviates from its expected value by an additive factor of $cn$ is exponential in $\Theta(n)$.
\end{abstract}

\section{Introduction}
The Bucket Sort algorithm sorts $n$ keys in the interval $[0,1)$ as
follows:
\begin{enumerate*}[label={(\roman*)}]
\item Distribute the keys among $n$ buckets, where the $j$th
  bucket consists of all the keys in the interval $[j/n,(j+1)/n)$.
\item Sort the keys in each bucket.
\item Scan the buckets and output the keys in each bucket in their
  sorted order.
\end{enumerate*}
We consider two natural classes of Bucket Sort algorithms that differ
in how the keys inside each bucket are sorted. The first class of
BucketSort algorithms that we consider sorts the keys inside a bucket
using a quadratic time algorithm (such as Insertion Sort). We refer to
algorithms in this class as \sbs. The second class of algorithms sorts
the keys in a bucket using a $\Theta(b\log b)$ algorithm for sorting $b$
keys (such as Merge Sort).  We refer to this variant as \mbs.

When the $n$ keys are distributed
independently and uniformly at random, the expected running time of Bucket Sort is $\Theta(n)$, 
even when a
quadratic time algorithm is used to sort the keys in each
bucket\cite{cormen2009introduction,mitzenmacher2017probability,sanders2019sorting}. 
A natural question is whether such linear time guarantees hold with high probability.
For Quick Sort, analyses of this sort have a long and rich history~\cite{janson2015tails,
  fill2002quicksort, mcdiarmid1996large}.

In this paper, we focus on analyzing the running time of Bucket Sort
with respect to large deviations, e.g., running times that exceed the expectation by $10n$. 
In particular, we study the asymptotic behavior
 of the exponent in the upper tail of the running time.

\medskip\noindent
\textbf{Rate of the upper tail.} We analyze the upper tail probability of a random variable using the notion of rate, defined as
follows.\footnote{Throughout the paper, $\ln x$ denotes the natural logarithm of $x$ and $\log x$ denotes the logarithm of base $2$ of $x$.} 
\begin{definition}  Given a random variable $Y$ with expected value $\mu$, we define
 the \emph{rate} of the upper tail of $Y$ to be 
  the function defined on $t>0$ as follows:
\begin{center}
$R_{Y}(t) \triangleq -\ln\parentheses{\Pr{Y \geq \mu + t}}.$
\end{center}
\end{definition}

Note that we consider an additive deviation from the expectation, i.e., we bound
the probability that the random variable deviates from its expected
value by an additive term of $t$, for sufficiently large values of
$t$.\footnote{One should not confuse this analysis with concentration
  bounds that address small deviations from the expectation.} In
particular, we consider values of $t=cn$, where $n$ is the size of the
input and $c$ is a constant greater than some threshold. Finally, we
abbreviate and refer to $R_Y(t)$ as the \emph{rate} of $Y$.

We study the rates of the running times of deterministic Bucket Sort algorithms  in which the input
is sampled from a uniform probability distribution. We also consider parameters of tries induced by infinite
prefix-free binary strings chosen independently and uniformly at random.

\subsection{Our Contributions}

Our first two results derive the rates of the two classes of Bucket Sort algorithms 
 and show that they are different. Specifically, we prove the following: 

\begin{theorem}\label{thm:k2}
  There exists a constant $C>0$ such that, for all $c>C$, the rate
  $R_{b^2}(\cdot)$ of the \sbs algorithm on $n$ keys chosen
  independently and uniformly at random in $[0,1)$ satisfies
  $R_{b^2}(cn) = \Theta(\sqrt{n}\log n)$.
\end{theorem}

Since the expected running time of $\sbs$ is $\Theta(n)$,
Theorem~\ref{thm:k2} states that the probability that \sbs 
on random keys takes more than $dn$ time is
$e^{-\Theta(\sqrt{n}\log n)}$ (for a suffciently large constant
$d$).\footnote{The threshold $C$ depends on: (1)~the constant that
  appears in the sorting algorithm used within each bucket, and
  (2)~the constant that appears in the expected running time of \sbs.}
Theorem~\ref{thm:k2} proves both a lower bound and an upper bound on
the asymptotic rate $R_{b^2}(cn)$.  In particular,
Theorem~\ref{thm:k2} rules out the possibility that the probability
that the running time of \sbs is greater than $100n$ is
bounded by $e^{-\Theta(n)}$.

We prove the lower bound on $R_{b^2}(cn)$ by
applying multiplicative Chernoff bounds in different regimes of large (superconstant, in fact)
deviations from the mean.  In such settings, the dependency of the exponent of the Chernoff
bound on the deviation from the mean can have a significant impact on the quality of the bounds we obtain. 
Indeed, we employ a rarely used form of the Chernoff bound that exhibits a $\delta \log \delta$ dependency in the exponent
when the deviation from the mean is $\delta$ (see Eq.~\ref{eq:ch4} in Appendix~\ref{app:chernoff} and Chapter $10.1.1$ in~\cite{doerr2018tools}). Although the proof of this bound is straightforward, the proof of
Theorem~\ref{thm:k2} crucially relies on this additional (superconstant) $\log \delta$ factor (see Claim~\ref{claim:sum3}).

\medskip\noindent
For \mbs on 
random keys, we show that the rate is linear in the size of the input:
\begin{theorem} \label{thm:klogk}There exists a constant $C>0$ such
  that, for all $c>C$, the rate $R_{b\log b}(\cdot)$ of the \mbs
  algorithm on $n$ keys chosen independently and uniformly at random
  in $[0,1)$ satisfies $R_{b\log b}(cn) = \Theta(n)$.
\end{theorem}

We prove the lower bound on $R_{b\log b}(cn)$ by analyzing a random variable arising in random tries. Specifically, we consider tries on infinite binary strings in which each bit is chosen independently and uniformly at random. The parameter we study is called the \emph{excess path length} and is defined formally in Section~\ref{sec:prelim}. We show that the time it takes to sort the buckets in \mbs can be upper bounded by the excess path length in a random trie (Lemma~\ref{lem:red}). We then bound the upper tail of the excess path length (Theorem~\ref{thm:excess}) and use it to lower bound $R_{b\log b}(cn)$. 

We also use the upper tail of the excess path length to derive the rate of a well-studied trie parameter, the sum of root to leaf paths in a minimal trie, called the \emph{nonvoid external path length}~\cite{Knuth,sedgewick2013introduction}. It is known that the expected value of 
the nonvoid external path length in a random trie is $n\log n + \Theta(n)$~\cite{Knuth,
  szpankowski2011average, sedgewick2013introduction}. We show the following:

\begin{theorem}\label{thm:external}
  There exists a constant $C>0$ such that, for all $c>C$, the rate
  $R_0(\cdot)$ of the nonvoid external path length of a minimal
  trie on $n$ infinite binary strings chosen independently and
  uniformly at random satisfies $R_0(cn) = \Theta(n)$.
\end{theorem}

Note that Theorem~\ref{thm:external}  implies that the probability that the nonvoid external
path length is more than $n\log n + dn$ is $e^{-\Theta(n)}$ (for a sufficiently large constant $d$).

\subsection{Related Work}
Showing that Bucket Sort runs in linear expected time when the keys are distributed independently and 
uniformly at random in $[0,1)$ is a classic textbook result~\cite{cormen2009introduction,mitzenmacher2017probability,sanders2019sorting}.
Bounds on the expectation as well as limiting distributions for the running time
have also been studied for different versions of Bucket Sort~\cite{mahmoud2000analytic,devroye1986lecture}. We are not
aware of any work that directly addresses the rate of the running time of Bucket Sort. 
The upper and lower tails of the running time of Quick Sort  have been studied in depth~\cite{janson2015tails, fill2002quicksort}, including
in the regime of large deviations~\cite{mcdiarmid1996large}.

The expected value of the nonvoid external path length of a trie 
is a classic result in applying the methods of analytic combinatorics to the analysis of algorithms~\cite{Knuth, 
  szpankowski2011average, mahmoud1992evolution, sedgewick2013introduction,clement2001dynamical}.
We consider the case in which the binary strings are independent and
random (i.e., the bits are independent and unbiased). In~\cite{Knuth,
  szpankowski2011average, mahmoud1992evolution, sedgewick2013introduction} it is shown that
for random strings, the expected value of the nonvoid external
path length is $n\log n + \Theta(n)$. The variance of the nonvoid external path length and
limiting distributions for it have also been studied extensively for
different
string distributions~\cite{jacquet1988normal,kirschenhofer1989variance,DBLP:books/el/leeuwen90/VitterF90}.

In Knuth~\cite[Section 5.2.2]{Knuth}, the nonvoid external path length is shown
to be proportional to the number of bit comparisons of radix exchange sort. 
The bound in Thm.~\ref{thm:external} therefore applies to the rate
of the number of bit comparisons of radix exchange sort when the strings are
distributed independently and uniformly at random. 

The connection between the running time of sorting algorithms and various trie parameters (including external path length) has also been studied by Seidel~\cite{seidel2010data}, albeit in a significantly different model than ours.
 Specifically,~\cite{seidel2010data} analyzes the expected number of bit comparisons of Quick Sort and Merge Sort when the
input is a randomly permuted set of strings sampled from a given distribution. In Seidel's model, the cost of comparing two
strings is proportional to the length of their longest common prefix. Seidel shows that the running time of these algorithms
can be naturally expressed in terms of parameters of the trie induced by the input strings. We emphasize that our
analysis connects the running time of Bucket Sort to the excess path length in the comparison model (in
which the cost of comparing two keys does not depend on their binary representation).
 
\subsection{Paper Organization}
Preliminaries and definitions are in Sec.~\ref{sec:prelim}. In Section~\ref{sec:red}, we present reductions from the running time of \mbs and the nonvoid external path length to the excess path length. The bound on the upper tail of the excess path length is proved in Sec.~\ref{sec:lbexcess}. Section~\ref{sec:lbsbs} proves a lower 
bound on the rate of \sbs. Upper bounds on the rates are proved in Appendix~\ref{app:ub}. Theorems~\ref{thm:k2},~\ref{thm:klogk} and~\ref{thm:external} are completed in Sec.~\ref{sec:final}. 
Finally, in Sec.~\ref{sec:discussion}, we include a discussion on the difference between the rate of Bucket Sort
and that of Quick Sort.

\section{Preliminaries and Definitions}\label{sec:prelim}
\textbf{Bucket Sort.} The input to Bucket Sort consists of $n$ keys $\bX\triangleq \set{x_1,\ldots,x_n}$
in the interval $[0,1)$.  We define bucket $j$ to be the set of keys
in the interval $[j/n,(j+1)/n)$. Let $\bB(\bX)\triangleq (B_0,\ldots B_{n-1})$ be
the \emph{occupancy vector} for input $\bX$, where $B_j$ denotes the number of
keys in $\bX$ that fall in bucket $j$.

The buckets are separately sorted and the final output is computed by
scanning the sorted buckets in increasing order.  The initial
assignment of keys to buckets and the final scanning of the sorted
buckets takes $\Theta(n)$ time. We henceforth focus only on the time spent
on sorting the keys in each bucket.

We consider the two natural options for sorting buckets:
\begin{enumerate*}[label=(\roman*)]
\item Sort $b$ keys in time $\Theta(b^2)$, using a sorting
  algorithm such as Insertion Sort or Bubble Sort. We refer to this option as \sbs. 
\item Sort $b$ keys in time $\Theta(b\log b)$ using a sorting algorithm such as Merge Sort or Heap Sort. We refer to this option as \mbs. 
\end{enumerate*}
Let $[n]$ denote the set $\set{0,\ldots, n-1}$ and let
$\bB=(B_0,\ldots,B_{n-1})$ denote an arbitrary occupancy vector. We define the functions
\begin{center}
  $f(\bB) \triangleq  \sum_{j\in[n]} B^2_j$ \hspace{1.5cm} $g(\bB) \triangleq  \sum_{j\in[n], B_j>0} B_j\log B_j$.
\end{center}
We let $T_{b^2}(\bX)$ and $T_{b\log b}(\bX)$ denote the running time
on input $\bX$ of \sbs and \mbs, respectively.  Then,
$T_{b^2}(\bX)=\Theta(n+f(\bB(\bX))$ and
$T_{b\log b}(\bX)=\Theta(n+g(\bB(\bX))$.~\footnote{Interestingly, the
  sum of squares of bin occupancies, i.e., $f(\bB)$, also appears in
  the FKS perfect hashing construction~\cite{fredman1982storing}.}

  \medskip\noindent \textbf{Excess Path Length and Tries.}  We let
  $\size{\alpha}$ denote the length of a binary string
  $\alpha \in \set{0,1}^*$.  For a set $L$, let $\size{L}$ denote the
  cardinality of $L$.

\begin{definition}
  A set of strings $\set{\alpha_1,\ldots, \alpha_s}$ is
  \emph{prefix-free} if, for every $i\neq j$, the string $\alpha_i$ is not
  a prefix of $\alpha_j$.
\end{definition}

A \emph{trie} is a rooted binary tree with edges labeled $\set{0,1}$ such that two edges emanating from the same trie node are labeled differently.  For a binary string $\alpha$, let $\pi(\alpha)$ denote
the trie node $v$, where the path from the root to $v$ is labeled
$\alpha$.  We say that a trie node $u$ is a \emph{predecessor} of $v$
if $u$ is in the path from the root to $v$.  For a set $U$ of trie
nodes, the reduced trie that contains $U$ is denoted by $T(U)$,
namely, $T(U)$ consists of $U$ and all the predecessors of nodes in
$U$.  Given a set of binary strings $L$, let $T(L)$ denote the trie
$T(\pi(L))$. If the set $L$ is prefix-free and contains only finite-length strings, then every node in $\pi(L)$ is a leaf of $T(L)$.

For a set $L$ of prefix-free binary strings, let $\varphi_0(L)$ denote
the set of minimal prefixes of strings in $L$ subject to the
constraint that $\varphi_0(L)$ is prefix-free. The trie $T(\varphi_0(L))$
is called the \emph{minimal} trie on $L$. Note that the structure
of $\varphi_0(L)$ (or of $T(\varphi_0(L))$) does not change if we append more bits to the
strings in $L$.

The following definition extends the definition of $\varphi_0(L)$ by
requiring that the prefixes have length at least $k$.
\begin{definition}[minimal $k$-prefixes]\label{def:min k prefix}
  Let $L= \set{\alpha_0,\ldots,\alpha_{n-1}}$ denote a set of $n$
  distinct infinite binary strings.  Given a parameter $k\geq 0$, the
  set of\emph{minimal k-prefixes of $L$}, denoted by
  $\varphi_k(L)\triangleq \set{\beta_0,\ldots, \beta_{n-1}}$, is the set that
  satisfies the following properties:
 \begin{enumerate}
\item for all $i\in[n]$, the string $\beta_i$ is a prefix of $\alpha_i$, 
\item for all $i\in[n]$, $\size{\beta_i} \geq k$,  
\item The set $\varphi_k(L)$ is prefix-free,
\item $\displaystyle{\sum_{i=0}^{n-1} \size{\beta_i}}$ is minimal among
  all sets that satisfy the first $3$ conditions.
\end{enumerate}
\end{definition}

The embedding of $\varphi_k(L)$ in a trie maps every string in
$\varphi_k(L)$ to a distinct leaf of depth at least $k$.
The definition can be modified to handle prefix-free sets of finite
strings by appending an arbitrary infinite string (say,
zeros) to each finite string. In this
paper, we are interested in the following trie parameter defined on $\varphi_k(L)$:

\begin{definition}
The \emph{$k$-excess path
    length} $p_k(L)$ of a set $L$ of $n$ distinct infinite binary strings is defined as:
\begin{center}
  $p_k(L)\triangleq \sum_{\alpha\in\varphi_k(L)} \parentheses{\size{\alpha} - k}$.
\end{center} 

\end{definition}

In~\cite{sedgewick2013introduction}, $p_0(L)$ is called the
\emph{nonvoid external path length} of the minimal trie on $L$. When $k=\ceil{\log \size{L}}$, we simply refer to $p_k(L)$ as the \emph{excess
  path length} of $L$.

\medskip\noindent \textbf{Distributions.}  Let $\mathcal{X}_n$ denote
the uniform distribution over $[0,1)^n$. Note that if the set
$\bX = \set{x_0,\ldots,x_{n-1}}$ is chosen according to
$\mathcal{X}_n$, then $x_0,\ldots,x_{n-1}$ are chosen independently
and uniformly at random from the interval $[0,1)$. Let $\mu_{b^2}$
(res. $\mu_{b\log b}$) denote the expected ruuning time $T_{b^2}(X)$
(resp., $T_{b\log b}(X)$) when $X\sim\mathcal{X}_n$. Similarly, let
$\mu_{f}$ (res. $\mu_{g}$) denote the expected values of $f(\bB(X))$
(resp., $g(\bB(X))$) when $X\sim\mathcal{X}_n$. It is known that
$\mu_f = 2n-1$
(see~\cite{cormen2009introduction,mitzenmacher2017probability}), and
consequently, we have that $\mu_{b^2} = \Theta(n)$.  Since $g\leq f$,
we also have $\mu_g= \Theta(n)$ as well as
$\mu_{b\log b} = \Theta(n)$.

Let $\mathcal{L}_n$ denote the uniform distribution over $n$ infinite
binary strings. Note that if $\bL = \set{\alpha_0,\ldots,\alpha_{n-1}}$ is chosen according
to $\mathcal{L}_n$, then all the bits of the strings are
independent and unbiased. We let $\mu_0$ denote the expected value of the external nonvoid path $p_0(L)$ when $L\sim\mathcal{L}_n$. It is know that $\mu_0=n\log n + \Theta(n)$ (see ~\cite{Knuth,
  szpankowski2011average, sedgewick2013introduction}).

\medskip\noindent \textbf{Rates.}
Let $R_{b^2}(\cdot)$ (resp., $R_{b\log b}(\cdot)$)
denote the rate of $T_{b^2}(X)$
(resp., $T_{b\log b}(X)$) when $X\sim\mathcal{X}_n$.
Similarly, let
$R_f(\cdot)$ (resp., $R_g(\cdot)$) 
denote the rate of $f(\bB(X))$
(resp., $g(\bB(X))$) when $X\sim\mathcal{X}_n$.

We first note that, to study the asymptotic behavior of $R_{b^2}$ (for
sufficiently large deviations) it suffices to study the asymptotic
behavior of $R_f$. The proof of the following appears in
Appendix~\ref{app:reduction}. 
\begin{observation}\label{obs:reduction} For every $c>0$, there exist constants
  $\delta_1 = \Theta(c)$ and $\delta_2=\Theta(c)$ such that:
\begin{center}
$R_f(\delta_1 \cdot n) \leq R_{b^2}(c\cdot n) \leq R_f(\delta_2\cdot n)$.
\end{center}
\end{observation}
An analogous statement holds for the rates $R_{b\log b}$ and $R_g$. The rate of the nonvoid external
path length $p_0(L)$ is denoted by $R_{0}(\cdot)$.

\section{Reductions}\label{sec:red}
\subsection{Balls-into-Bins Abstraction}
We interpret the assignment of keys to buckets using a balls-into-bins
abstraction.  The keys correspond to balls, and the buckets correspond
to bins.  The assumption that $\bX\sim\mathcal{X}_n$ implies that the
balls choose the bins independently and uniformly at random. The
value $B_j$ then equals the occupancy of bin $j$.

A similar balls-into-bins abstraction holds for the embedding of the
minimal $(\log n)$-prefixes of $L\sim\mathcal{L}_n$ in a trie
(assuming $n$ is a power of $2$).  Indeed, let
$\set{v_0,\ldots,v_{n-1}}$ denote the $n$ nodes of the trie $T(L)$ at
depth $\log n$.  For a node $v_j$, we say that a string $\alpha$
\emph{chooses} $v_j$, if the path labeled $\alpha$ contains $v_j$.
Since the strings are random, each string chooses a node of depth
$\log n$ independently and uniformly at random. Let $C_j$ denote the
number of strings in $L$ who choose node $v_j$.\footnote{Formally,
  $T(L)$ may contain a subset of these $n$ nodes. If a node $v_j$ at
  depth $\log n$ is not chosen by any string, then define $C_j=0$.} We
refer to $C_j$ as the \emph{occupancy} of $v_j$ with respect to $L$
and define the vector $\bC(L) \triangleq (C_0,\ldots, C_{n-1})$.

\begin{observation}\label{obs:same}
  When $\bX\sim \mathcal{X}_n$ and $L \sim \mathcal{L}_n$, the occupancy vector $\bB(\bX)$ 
  has the same joint probability distribution as  $\bC(L)$.
\end{observation}

\subsection{Lower Bounding the Rate of \mbs}
By Obs.~\ref{obs:reduction}, to prove a lower bound in
$R_{b \log b}$ it suffices to prove a lower bound on $R_g$. In this
section we show how to lower bound $R_g$ by bounding the upper tail
probability of the excess path length $p_{\log n}(L)$.  We begin with
the following observation about the nonvoid external path length
$p_0(L)$:
\begin{observation}\label{obs:p0}
  For every set $L$ of $n$ infinite prefix-free binary strings,
  $p_0(L)\geq n\log n$.  
\end{observation}

Now consider an arbitrary vector $\bC(L)$ and apply Observation~\ref{obs:p0} to each node of depth $\log n$
separately. We obtain the following corollary.
\begin{corollary}\label{cor:same} For every set $L$ of $n$ infinite
  prefix-free binary strings, $p_{\log n}(L) \geq g(\bC(L))$.
\end{corollary}

\medskip\noindent 
We lower bound the rate of $g(\bB(\bX))$ as
follows. 
\begin{lemma}\label{lem:red}
For every $c>0$, 
\begin{align}\label{eq:g vs p}
  \Prr{X\gets \mathcal{X}_n}{g(\bB(\bX)) \geq \mu_g+ cn} &\leq
\Prr{L\gets\mathcal{L}_n}{p_{\log n}(L) \geq c n} \;.
\end{align}
\end{lemma}
\begin{proof}
Recall that $\mu_g$ denotes the expected value of of $g(\bB(\bX))$. Since
$\mu_g>0$, we have that
$\Pr{g(\bB(X))\geq \mu_g+cn} \leq \Pr{g(\bB(X))\geq cn}$. 
Observation~\ref{obs:same} implies that
\begin{align*}
 \Prr{X\gets \mathcal{X}_n}{g(\bB(\bX)) \geq cn} &=
\Prr{L\gets\mathcal{L}_n}{g(\bC(L)) \geq c n}\;.
\end{align*}
The claim then follows by Corollary~\ref{cor:same}.
\end{proof}
\medskip\noindent
Hence, a lower bound on the rate  of $g(\bB(X))$ follows by bounding the RHS of Eq.~\ref{eq:g vs p}.

\subsection{Lower Bounding the Rate of the  Nonvoid External Path Length}
In this section, we show how to use the upper tail of $p_{\log n}(L)$ to lower bound the rate 
of the nonvoid external path length $p_0(L)$.
\begin{observation}\label{obs:trie}
For every set $L$ of $n$ infinite prefix-free binary strings, we have that:
\begin{center}
$p_0(L) \leq n\log n + p_{\log n}(L)$.
\end{center}
\end{observation}
\begin{proof}
  The strings in $\varphi_0(L)$ are themselves prefixes of strings in
  $\varphi_{\log n}(L)$. We therefore get that
  $\sum_{\alpha\in \varphi_0(L)} \size{\alpha} \leq \sum_{{\beta}\in
    \varphi_{\log n}(L)}\size{\beta}$, and the claim follows.
\end{proof}

\medskip\noindent
Observation~\ref{obs:p0}
implies that $\mu_0\geq n\log n$.
Together with Obs.~\ref{obs:trie} this implies that:
\begin{corollary}\label{cor:p0 vs p} For every $L\sim\mathcal{L}_n$ and every $c>0$.
\begin{center}
$\Pr{p_0(L)\geq \mu_0 + cn} \leq \Pr{p_{\log n}(L) \geq cn}$.
\end{center}
\end{corollary}

\section{The Upper Tail of the Excess Path Length}\label{sec:lbexcess}
We bound the upper tail of $p_{\log n}(L)$ as follows:
\begin{theorem}\label{thm:excess}
Let $L \sim \mathcal{L}_n$. For every $c>0$:
\begin{align*}
\Pr{p_{\log n}(L)\geq (8c + 16)\cdot n} &\leq  \exp\parentheses{-\frac{c-1-\ln c}{4}\cdot n }  \;.
\end{align*}
\end{theorem}

\begin{proof}

  Let $L=\set{\alpha_1,\ldots, \alpha_{n}}$ be a set of infnite random
  binary strings.  We consider the evolution of the set
  $\varphi_{\log n}(L)$ of minimal $\log n$-prefixes as we process the
  strings $\alpha_i$ one by one.  Specifically, let
  $L^{(i)} \triangleq \set{\alpha_1, \ldots, \alpha_{i}}$, for
  $1\leq i\leq n$, and $L_0 = \emptyset$.

\medskip\noindent
Let $\varphi(L^{(i)}) \triangleq \set{s_j \circ \delta_{j}^{(i)} \Bigm\vert s_j \circ \delta_{j}^{(i)} \text{is a prefix of } \alpha_j \text{ and } \size{s_i} = \ceil{\log n}  \text{ for } 1\leq j \leq i}$. Note that 
\begin{center}
 $p_{\log n}(L^{(i)})=\displaystyle{\sum_{j\in [i]} \size{\delta^{(i)}_j}}$ .
\end{center}

We bound $p_{\log n}(L)$ by considering the increase 
$\Delta_i \triangleq p_{\log n}(L^{(i)}) - p_{\log
  n}(L^{(i-1)})$. Since $p_{\log n}(L^{(0)}) = 0$ and
$p_{\log n}(L^{(n)}) = p_{\log n}(L)$, then
$p_{\log n}(L) = \sum_{i=1}^n \Delta_i$.

The addition of the string $\alpha_i$  has two types of contributions
to $\Delta_i$. The first contribution is $\delta_{i}^{(i)}$.
The second contribution is due to the need to extend colliding
strings.  Indeed, since the set $L^{(i-1)}$ is prefix-free,
there exists at most one $j<i$ such that $s_j \circ \delta^{(i-1)}_j$
is a prefix of $\alpha_i$.
If $s_j \circ \delta^{(i-1)}_j$
is a prefix of $\alpha_i$, then 
$\Delta_i = \size{\delta^{(i)}_j} - \size{\delta^{(i-1)}_j} +
\size{\delta^{(i)}_i}$. Because $\delta^{(i)}_j$ and $\delta^{(i)}_i$
are minimal subject to being prefix-free, we also have that
$\size{\delta^{(i)}_j} = \size{\delta^{(i)}_i}$. Hence,
$\Delta_i \leq 2\cdot \size{\delta^{(i)}_i}$. This implies that, for
every $\tau$:
\begin{align*}
\Pr{\Delta_i \geq 2\tau} &\leq\Pr{\size{\delta^{(i)}_i}\geq\tau}\;.
\end{align*}

\medskip\noindent
We now proceed to bound $\Pr{\size{\delta^{(i)}_i}\geq\tau}$. 
Fix $i\geq 1$ and let $\delta_i(\ell)$ denote the prefix
of length $\ell$ of $\delta^{(i)}_i$. We denote by $n_{\ell}$  the number of leaves in
the subtree rooted at $s_i \circ\gamma_i(\ell)$  in the trie $T(L^{(i-1)})$ (i.e., right before the string $\alpha_i$ is processed).  Formally,
\begin{align*}
n_\ell&\triangleq \size{\set{j<i \Bigm\vert s_i \circ \delta_i(\ell) \text{ is a prefix of } 
		s_j \circ \delta^{(i-1)}_j}}\;.
\end{align*}

Clearly,
$n_0 = \size{\set{j<i \bigm\vert s_i = s_j }}$ and
$n_{\size{\gamma^{(i)}_i}}=0$. We bound $\size{\gamma^{(i)}_i}$ by bounding
the minimum $\ell$ for which $n_\ell$ becomes zero as follows: define the binary random variable $Z_{\ell+1}$ to be $1$ if
$n_{\ell+1} \leq \frac{1}{2}\cdot n_\ell$, and $0$ otherwise. Note that $\Pr{Z_\ell=1}\geq 1/2$ and that $\set{Z_\ell}_\ell$ are
independent.
By
definition,
\begin{align}\label{eq:sum}
\size{\delta^{(i)}_i} \geq \tau 
~\Longrightarrow~
\sum_{s=1}^{\tau} Z_s \leq \log (1+n_0)\;.
\end{align}

\medskip\noindent
By the law of total probability,
\begin{align}\label{eq:alpha}
\Pr{\size{\delta^{(i)}_i} \geq \tau} 
&\leq
\Pr{\log(1+n_0) \geq \tau/8}+\Pr{\parentheses{\size{\delta^{(i)}_i} \geq \tau }\Bigm\vert \log (1+n_0) \leq \tau/8}\;.
\end{align}

\medskip\noindent
We now bound the two terms in the RHS of Eq.~\ref{eq:alpha}.  Note
that $\expectation{}{n_0} \leq 1$.  In fact
$\expectation{}{n_0 \bigm\vert \bigwedge_{j<i} \Delta_j=\xi_j}\leq 1$
for every realization $\set{\xi_j}_{j<i}$ of $\set{\Delta_j}_{j<i}$.
By Markov's inequality:
\begin{align}\label{eq:alpha1}
\Pr{1+n_0 \geq 2^{\tau/8}} &\leq \frac{1+\expectation{}{n_0}}{2^{\tau/8}}
\leq {2^{-\tau/8+1}}\;.
\end{align}
To bound the second term in the  RHS of Eq.~\ref{eq:alpha}, we apply the Chernoff bound in Eq.~\ref{eq:ch5}:
\begin{align}\label{eq:alpha2}
\nonumber \Pr{\size{\delta^{(i+1)}_i} \geq \tau \Bigm\vert \log (1+n_0) \leq \tau/8} & \leq
\Pr{\sum_{s=1}^{\tau} Z_s \leq \frac{\tau}{8}} \hspace{3cm} (\text{By Eq.~\ref{eq:sum}})\\
\nonumber&\leq \Pr{\sum_{s=1}^{\tau} Z_s \leq \frac{2}{8}\cdot\expectation{}{\sum_{s=1}^{\tau} Z_s } } \hspace{0.72cm} (\expectation{}{Z_i}\geq 1/2)\\
\nonumber&\leq \exp\parentheses{-\frac{1}{2}\cdot \expectation{}{\sum_{s=1}^{\tau} Z_s } \cdot \parentheses{1-\frac{2}{8}}^2}\\
&\leq \exp\parentheses{-\frac{\tau}{4}\cdot \parentheses{\frac{3}{4}}^2}
=\exp\parentheses{-\frac{9}{64}\cdot\tau}\;.
\end{align}
\noindent
From Equations~\ref{eq:alpha} -- \ref{eq:alpha2}, it follows that:
\begin{align*}
\Pr{\size{\delta^{(i+1)}_i} \geq \tau} &\leq  2^{-\tau/8+1} + \exp(-9\tau/64)\leq 2^{-\tau/8+2}\;.
\end{align*}
Therefore,
\begin{align}\label{eq:delta}
\Pr{\Delta_i \geq  16\cdot (\tau+2))}
&\leq 2^{-\tau}\;.
\end{align}
Note that Eq.~\ref{eq:delta}  also holds under every conditioning on
the realizations of $\set{\Delta_{j}}_{j<i}$. 

Let $\Delta'_i\triangleq \frac{1}{16}\cdot \Delta_i - 1$ and note that
$\Pr{\Delta'_i\geq \tau}\leq 2^{-\tau+1}$.  Let $\{G_i\}_i$
denote independent geometric random variables, where
$G_i\sim Ge(1/2)$.  Since $\Pr{G_i\geq \tau}=2^{-(\tau-1)}$, we
conclude that $\Delta'_i$ is stochastically dominated by $G_i$.  In
fact, the random variables $\set{\Delta'_i}_{i\in[f]}$ are
unconditionally sequentially dominated by $\set{G_i}_{i\in[n]}$.
By~\cite[Lemma 8.8]{doerr2018tools}, it follows that
$\sum_{i\in[n]} \Delta'_i$ is stochastically dominated by $\sum_{i\in [n]} G_i$.\footnote{Note that RVs $\set{\Delta}_i$ are not 
independent and probably not even negatively associated. Hence, standard concentration bounds do not apply to $\sum \Delta_i$.}

\medskip\noindent
The sum of independent geometric random variables is
concentrated~\cite{janson2018tail} and so we get:
\begin{align*}
\Pr{\sum_{i\in[n]} \Delta'_i \geq c\cdot n/2} &\leq \Pr{\sum_{i\in[n]} G_i \geq c \cdot n/2}\\
&\leq \exp\parentheses{-\frac{c-1-\ln c}{4}\cdot n}
\end{align*}
as required.
\end{proof}

\section{Lower Bound for \sbs}\label{sec:lbsbs}

This section deals with proving the following lower bound on the rate
$R_f$.  By Obs.~\ref{obs:reduction}, this also implies a lower bound
on the rate $R_{b^2}$.

\begin{lemma}\label{lem:lowersbs}
There exists a constant $C>0$ such that, for all $c>C$, we have that $R_{f}(cn) = \Omega(\sqrt{n}\log n)$, for all sufficiently large $n$.
\end{lemma}

\subsection{Preliminaries}

Given an input $\bX$ of $n$ keys and its associated occupancy vector
$\bB(\bX) = (B_0, B_1, \ldots, B_{n-1})$, define
$\calS_i \triangleq \set{j \in [n]\mid B_j \geq i}$ to be the set of
buckets with at least $i$ keys assigned to them.  Note that the random variables $\set{\size{\calS_i}}_i$
are negatively associated because they are monotone functions of bin occupancies, which are a classical
example of negatively associated RVs~\cite{dubhashi2009concentration}.

\begin{claim2}
\label{claim:double count}
  For every occupancy vector $(B_0,B_1, \ldots, B_{n-1})$, the following holds:
\begin{align}\label{eq:s1} 
  \sum_{j\in[n]} \binom{\size{B_j}+1}{2} &= \sum_{i\in[n+1]} i \cdot \size{\calS_i}\;.
  \end{align}
\end{claim2}
\begin{proof}
  Consider an $n\times n$ matrix $A$ filled according to the following rule: 
  \begin{align*}
    A_{i,j} &\triangleq
              \begin{cases}
                i &\text{if $B_j\geq i$}\\
                0 &\text{otherwise.}
              \end{cases}
  \end{align*}
  Let $S\triangleq \sum_{i,j} A_{i,j}$.  The sum of entries in column
  $j$ equals $\binom{\size{B_j}+1}{2}$.   On the other hand, the
  sum of entries in row $i$ equals $i\cdot \size{\calS_i}$. Hence
  both sides of Eq.~\ref{eq:s1} equal $S$, and the claim follows.
\end{proof}

\medskip\noindent
Lemma~\ref{lem:isi} states that, in
order to prove Lemma~\ref{lem:lowersbs}, it suffices to prove a
lower bound on the upper tail probability of the random
variable $\sum_{i\in[n+1]} i\cdot \size{\calS_i}$. Specifically, , we
get that:

\begin{lemma}\label{lem:isi} For every $c$, we have that
\begin{align*}
\Pr{f(\bB(X)) \geq \mu_f + cn} &= \Pr{\sum_{i\in[n+1]} i\size{\calS_i} \geq  \frac{(3+c)n-1}{2}}\;.
\end{align*}
\end{lemma}
\begin{proof}
  By Claim~\ref{claim:double count},
  $ f(\bB(X)) = 2\cdot\sum_{i\in[n+1]} i \size{\calS_i}-n$. The Lemma
  follows from the fact that
  $\mu_f =
  2n-1$~\cite{cormen2009introduction,mitzenmacher2017probability}.
\end{proof}

\medskip\noindent
Next, we upper bound $\expectation{}{ \size{\calS_i}}$. Let
$E_i \triangleq \parentheses{\frac{e}{i}}^{i}$ and note the following:
 
\begin{claim2}\label{expectation-bucket} For
every $i\in \set{1,\ldots, n}$, we have that $\expectation{}{ \size{\calS_i}} \leq n \cdot E_i$.
\end{claim2} 
\begin{proof} Fix $i$ and let $X_{i,j}$ be the indicator random
variable that is $1$ if $B_j \geq i$ and $0$ otherwise. We get that
$\size{\calS_i} = \sum_j X_{i,j}$. Because each key chooses a bucket
independently and uniformly at random, we have that:
\begin{align*}
  \Pr{B_{j} \geq i} &
                      \leq {n \choose i}\cdot \parentheses{\frac{1}{n}}^{i}
                      \leq \parentheses{\frac{en}{i}}^{i}\cdot \parentheses{\frac{1}{n}}^{i} =\parentheses{\frac{e}{i}}^{i} = E_i\;.
\end{align*}
\medskip\noindent
The claim follows by linearity of expectation.
\end{proof}

\medskip\noindent
One can analytically show that:
\begin{observation}\label{obs:iei}
  $\sum_{i=1}^\infty i\cdot E_i \leq 10$.
\end{observation}

\subsection{Applying Chernoff Bounds in Different Regimes}
In the proof of Lemma~\ref{lem:lowersbs}, we consider three thresholds
on bin occupancies 
$\tau_1 \leq \tau_2 \leq\tau_3$ defined as follows:
\begin{align*}
  \tau_1\triangleq \max\set{i \Big\vert E_i\geq \frac{c\log n}{\sqrt{n}}}\;,
  \hspace{1.5cm} \tau_2 \triangleq \frac{n^{1/4}}{\sqrt{\log n}} \;, 
  \hspace{1.5cm}\tau_3 \triangleq \sqrt{n}\;.
\end{align*}

\begin{claim2}\label{claim:sum1} For every $c>0$, there exists a $\gamma=\gamma(c)>0$, such that:
\begin{align*}
  \Pr{\sum_{i\leq \tau_1} i\cdot \size{\calS_i} \geq cn+ \sum_{i\leq \tau_1} i \cdot E_i\cdot n} &\leq \exp\parentheses{-\gamma\sqrt{n}\log n}\;.
\end{align*}
\end{claim2}
\begin{proof}
  Fix $i\leq \tau_1$. By the Chernoff bounds
  (Eq.~\ref{eq:ch2}-\ref{eq:ch3} in Appendix~\ref{app:chernoff}) and
  the definition of $\tau_1$, for every $\delta>0$, there exists a
   $c'= c'(\delta)>0$, such that:
\begin{align*}
\Pr{\size{\calS_i} \geq (1+\delta)E_i\cdot n} &\leq \exp\parentheses{- c' \cdot E_i\cdot n}\leq \exp(-c'\cdot c\cdot \sqrt{n}\cdot \log n)\;.
\end{align*}
\medskip\noindent
By applying a union bound over all $i\leq \tau_1$, it follows that there
exists a $\gamma>0$ such that:
\begin{align*}
\Pr{\sum_{i\leq \tau_1} i\size{\calS_i} \geq (1+\delta)\cdot \sum_{i\leq \tau_1} i E_i\cdot n } &\leq \exp\parentheses{-\gamma\sqrt{n}\log n}\;.
\end{align*}
\medskip\noindent
Define $\delta \triangleq c/10$. By Obs.~\ref{obs:iei},
$\delta \sum_{i} iE_i \leq c$, and the claim follows.
\end{proof}

\begin{claim2}\label{claim:sum2} For every $c>0$, there exists a  $\gamma=\gamma(c)>0$ such that for $n$ sufficiently large:
\begin{align*}
\Pr{\sum_{i=\ceil{\tau_1}}^{\floor{\tau_2}} i\size{\calS_i} \geq cn + \sum_{i=\ceil{\tau_1}}^{\floor{\tau_2}} i E_i\cdot n} &\leq \exp\parentheses{-\gamma\sqrt{n}\log n} \;.
\end{align*}
\end{claim2}
\begin{proof}
  For every $\tau_1 (c)< i\leq \tau_2$,  define
  $\delta_i \triangleq (c\log n)/(E_i\sqrt{n})$ so that
\begin{align}\label{eq:c}
\sum_{i\leq \tau_2} \delta_i\cdot i  E_i &= \sum_{i\leq \tau_2} i \cdot \frac{c\log n}{\sqrt{n}}  
\leq (\tau_2)^2 \cdot \frac{c\log n}{\sqrt{n}} =c \;.
\end{align}
\medskip\noindent
Since $\delta_i >1$ for every $i>\tau_1$,  by the Chernoff bound in Eq.~\ref{eq:ch3}:
\begin{align*}
\Pr{\size{\calS_i} > (1+\delta_i) \cdot E_i \cdot n} &\leq \exp\parentheses{-\delta_i \cdot n\cdot E_i/3}
  \\
                                                     &= \exp\parentheses{-c/3 \cdot\sqrt{n}\log n}\;. 
\end{align*}
By applying a union bound over all $\tau_1 \leq i \leq \tau_2$, it follows
that there exists a constant $\gamma>0$ such that:
\begin{align} \label{eq:tau2}
\Pr{\sum_{i=\ceil{\tau_1}}^{\floor{\tau_2}} i\size{\calS_i} \geq  \sum_{i=\ceil{\tau_1}}^{\floor{\tau_2}} (1+\delta_i)\cdot iE_i\cdot n}&\leq \exp\parentheses{-\delta\sqrt{n}\log n}\;.
\end{align}
The claim follows by Eq.~\ref{eq:c} and~\ref{eq:tau2}.
\end{proof}

\begin{claim2}\label{claim:sum3}  For every $c>0$, there exists a $\gamma=\gamma(c)>0$ such that for sufficiently large $n$, we have that: 
\begin{align*}
\Pr{\sum_{i=\ceil{\tau_2}}^{\floor{\tau_3}} i\size{\calS_i} > cn + \sum_{i=\ceil{\tau_2}}^{\floor{\tau_3}} i E_i\cdot n  } &\leq \exp\parentheses{-\gamma\sqrt{n}\log(n)}\;.
\end{align*}
\end{claim2}
\begin{proof}
For every $\tau_2 \leq i \leq \tau_3$,  define $\delta_i \triangleq \frac{c}{5} \cdot \frac{\log n}{i\log i \cdot E_i\cdot \sqrt{n}}$ so that the following holds for sufficiently large $n$:
\begin{align}
  \sum_{i=\ceil{\tau_2}}^{\floor{\tau_3}} \delta_i\cdot i E_i &
                                                                =\frac{c}{5} \cdot \parentheses{\sum_{i=\ceil{\tau_2}}^{\floor{\tau_3}} \frac{1}{\log i}} \cdot \frac{\log n}{\sqrt{n}}\\
&\leq \frac{c}{5}\cdot \frac{\tau_3}{\log \tau_2} \cdot \frac{\log n}{\sqrt{n}} 
\leq  \frac{c}{5} \cdot \frac{\log n}{0.25 \log n - 0.5\log \log n} \leq c \;.\label{eq:c2}
\end{align}

\medskip\noindent For a sufficiently large $n$, it holds that
$\delta_i >1$; moreover $\log \delta_i \geq \Omega(i \log i)$ for every
$\tau_2 \leq i \leq \tau_3$. By the Chernoff bound
in Eq.~\ref{eq:ch4}:
\begin{align*}
  \Pr{\size{\calS_i} > (1+\delta_i)\cdot n\cdot E_i}
  &\leq  \exp\parentheses{{-\delta_i\ln(\delta_i)\cdot n \cdot E_i/2}}
    \leq \exp\parentheses{-\Omega(\sqrt{n} \log n)}\;.
\end{align*}

\medskip\noindent
By applying a union bound over all $\tau_2 \leq i \leq \tau_3$, it follows that there exists a $\delta(c)>0$ such that:
\begin{align}\label{eq:tau3}
\Pr{\sum_{i=\ceil{\tau_2}}^{\floor{\tau_3}} i\size{\calS_i} > \sum_{i=\ceil{\tau_2}}^{\floor{\tau_3}} (1+\delta_i)\cdot i\cdot E_i\cdot n}&\leq \exp\parentheses{-\delta\sqrt{n}\log n}\;.
\end{align}
The claim follows by Eq.~\ref{eq:c2} and~\ref{eq:tau3}.
\end{proof}

\begin{claim2} \label{claim:sum4} For every $c>0$, there exists a
   $\gamma=\gamma(c)>0$ such that for sufficiently large $n$,
  we have that:
\begin{align*}
\Pr{\sum_{i=\ceil{\tau_3}}^{n} i\size{\calS_i} > cn} &\leq \exp\parentheses{-\gamma\sqrt{n}\log n} \;.
\end{align*}
\end{claim2}
\begin{proof}
We apply Markov's inequality and get that there exists a $c'>0$ such that:
\begin{align*}
  \Pr{\size{\calS_i} \geq \frac{c}{i}} &\leq \frac{i\cdot E_i\cdot n}{c}\leq \exp\parentheses{-c'\sqrt{n}\log n}\;,
\end{align*}
where the last inequality holds because $i\geq \tau_3$.
The claim follows by applying a union bound over $i\geq \tau_3$.
\end{proof}

\subsection{Proof of Lemma~\ref{lem:lowersbs}}
In order to prove Lemma~\ref{lem:lowersbs}, we apply a union bound over the Claims~\ref{claim:sum1}--\ref{claim:sum4} and use Obs.~\ref{obs:iei}. It follows that for every $c\geq 0$, there
exists a $\gamma=\gamma(c)>0$ such that:
\begin{align*}
\Pr{\sum_{i=1}^n i\size{\calS_i} > (10+c)\cdot n } &\leq\exp\parentheses{-\gamma\sqrt{n}\log n}\;.
\end{align*}
Lemma~\ref{lem:lowersbs} then follows from  Lemma~\ref{lem:isi}.

\section{Proof of Theorems~\ref{thm:k2},~\ref{thm:klogk} and~\ref{thm:external}}\label{sec:final}

To prove Theorems \ref{thm:k2} and~\ref{thm:klogk}, we employ Obs.~\ref{obs:reduction} that
shows a reduction from $R_f$ (and $R_g$, respectively) to $R_{b^2}$ (and $R_{b\log b}$ respectively).
The lower bound for $R_f$ is discussed in Lemma~\ref{lem:lowersbs}. The lower bound for $R_g$ follows
from Lemma~\ref{eq:g vs p} and Theorem~\ref{thm:excess}. The lower bound for $R_0$ follows from Cor.~\ref{cor:p0 vs p} and Theorem~\ref{thm:excess}. Finally, we apply Lemma~\ref{lem:ub} to get matching upper bounds on $R_f,R_g$ and $R_0$.

\section{Discussion: Comparison to Quick Sort}~\label{sec:discussion}
We note that the rate of the running time of Quick Sort is smaller
than that of Bucket Sort~\cite{mcdiarmid1996large}.  Here, we refer to
the version of Quick Sort that picks a pivot $x$ uniformly at random
and then recurses on two subsets: the set of elements smaller
than $x$ and the set of elements greater than $x$. Let $T_{qs}(n)$ be
the number of comparisons that Quick Sort makes on $n$ randomly
permuted distinct keys. The expectation of $T_{qs}(n)$ is denoted by
$\mu_{qs}$ and equals $\Theta(n\log n)$. McDiarmid and
Hayward~\cite{mcdiarmid1996large} prove that for
$\frac{1}{\ln n}< \eps\leq 1$:
\begin{align*}
  \Pr{|T_{qs}(n)-\mu_{qs}| \geq \eps \mu_{qs}} =
  n^{-2\eps(\ln\ln n - \ln(1/\eps)+O(\log\log\log n))}\;.
\end{align*} 
Setting $\eps=c/\ln n$ (for $c>1$) implies that the rate $R_{qs}$ of
$T_{qs}(n)$ satisfies
$R_{qs}(cn) = O(\log \log \log n)$.

One may wonder why the upper tails of Quick Sort and Bucket Sort exhibit
different rates. We provide some intuition by examining the
distributions of occupancies induced by Quick Sort and Bucket Sort on
nodes of depth $\log n$ in a complete binary tree. Consider
occupancies defined by the Quick Sort recursion tree as follows. In
each recursive call, the pivot ``stays'' in the inner node, and the
two lists are sent to the left and right children. Hence, every node
is assigned a (possibly empty) list of keys.  We refer to the
distribution of occupancies across the $n$ nodes of depth $\log n$
as the \emph{Quick Sort distribution}.

The number of comparisons $T_{qs}(n)$ is bounded by $n\log n$ (a bound
on the number of comparisons until level $\log n$) plus the
comparisons starting from level $\log n$. Clearly, the number of
comparisons starting from level $\log n$ depends on the Quick Sort
distribution.

The Quick Sort distribution is very far from the distribution of
$\bB(X)$ when $X\sim\mathcal{X}_n$ (i.e., the occupancy vector in
Bucket Sort when the $n$ keys are distributed uniformly at random).
Specifically, consider the event $Z$ that the occupancies of the $n/2$
nodes of depth $\log n$ in the left subtree are all zeros.  In the
Quick Sort distribution, the probability of event $Z$ is at least
$1/n$, e.g., $Z$ occurs if the first pivot is the smallest element.
  In the Bucket Sort distribution, the probability of event $Z$ is
$2^{-n}$ (i.e., all the keys are in the interval $(1/2,1)$).

\section*{Acknowledgments}
We thank Seth Pettie for useful discussions.

\bibliography{main}
\appendix

\section{Upper Bounds}\label{app:ub}

\begin{lemma}\label{lem:ub}
For every constant $c>0$, the following hold:
\begin{enumerate}
\item $R_{f}(cn)= O(\sqrt{n}\log n)$,
\item $R_{g}(cn) = O(n)$, and
\item $R_0(cn) = O(n)$.
\end{enumerate} 
\end{lemma}
\begin{proof}
  To prove statements 1 -- 2, we define $\calA_{i,j}$ to be the event
  that $B_j = i$ (occupancy of bin $j$ equals $i$) and note that:
\begin{align*}
\Pr{\calA_{i,j}}&= {n \choose i}\cdot \parentheses{\frac{1}{n}}^i\cdot \parentheses{1-\frac{1}{n}}^{n-i}
        \geq \parentheses{\frac{n}{i}}^i \cdot \parentheses{\frac{1}{n}}^i\cdot \parentheses{1-\frac{1}{n}}^{n}\\
        &\geq \frac{1}{4}\cdot \parentheses{\frac{1}{i}}^i = 2^{-2-i\log i}\;.
\end{align*}
\medskip\noindent
Fix a bin $j$. It follows that:
\begin{align*}
\Pr{f(\bB(\bX))\geq cn}&\geq \Pr{\calA_{\sqrt{cn},j}} \geq \exp\parentheses{-\Omega(\sqrt{n}\log n)}\;.
\end{align*}
This implies that $R_{f}(cn)= O(\sqrt{n}\log n)$ for every constant $c>0$.

\medskip\noindent
For $i$ such that $i\log i =\Omega(n)$ (i.e., $i=\Theta(n/\log n)$), we have that:
\begin{align*}
\Pr{g(\bB(\bX))=\Omega(n)}&\geq \Pr{\calA_{i,j}} \geq  \exp\parentheses{-\Omega(n)}\;.
\end{align*}
This implies that $R_{g}(cn)=O(n)$ for every constant $c>0$.

\medskip\noindent Now we consider lower bounding
$\Pr{p_0(L) \geq \mu_0 + cn}$ for every $c>0$ (statement 3). Consider
the event $\mathcal{A}$ in which the set $L$ contains two binary
strings $\alpha_1$ and $\alpha_2$ that share a common prefix of length
$(\frac{c}{2}+1)\cdot n$. When $L\sim\mathcal{L}_n$, we have that
$\Pr{\mathcal{A}} \geq 2^{-(c/2+1)n}$.

On the other hand, we have that, if event $\mathcal{A}$ happens, then
the depth of nodes $\pi(\alpha_1)$ and $\pi(\alpha_2)$ is more than
$(c/2+2)\cdot n$ in the trie $T(\varphi_0(L))$ (i.e., at least $c/2+2$ bits are required
to separate $\alpha_1$ and $\alpha_2$). For the rest $n-2$ binary
strings, we use Obs.~\ref{obs:p0} and get that we needed at least
$(n-2)\log(n-2)$ bits to separate them. Since
$\mu_0 \leq n\log n + 2n$, we get that event $\mathcal{A}$ implies
that
\begin{align*}
p_0(L\mid \mathcal{A})\geq 2\cdot \parentheses{\frac{c}{2}+2}\cdot n+(n-2)\log (n-2) \geq n\log n +2n + cn \geq \mu_0 + cn\;.
\end{align*} 
In other words,
$\Pr{p_0(L) \geq \mu_0 + cn}\geq \Pr{A}\geq 2^{-(c/2+1)n}$, hence
$R_0(cn) = O(n)$.
\end{proof}

\section{Proof of Observation~\ref{obs:reduction}}\label{app:reduction}
\begin{proof}
  Consider the first inequality that we need to prove:
  \begin{align}\label{eq:first}
R_f(\delta_1 n)&\leq R_{b^2}(cn) \;.
  \end{align}
 
Recall that $\mu_f= \expectation{}{f(\bB(X))}=2n-1$, and
that $\mu_{b^2}= \expectation{}{T_{b^2}(X)}$.
Equation~\ref{eq:first} is equivalent to the inequality
(where $X\sim\mathcal{X}_n$).
  \begin{align*}
    \Pr{ T_{b^2}(X)\geq \mu_{b^2}+cn} \leq
    \Pr{ f(\bB(X))\geq \mu_{f}+\delta_1 n}\;.
  \end{align*}
  
  Recall that $T_{b^2}(\bX)=\Theta(n+f(\bB(\bX)))$.
  Let $0<c_1\leq c_2$ be constants such that, for
  sufficiently large values of $n$,
  \begin{align*}
    c_1 \cdot (n + f(\bB(\bX)) ) &\leq T_{b^2}(\bX) \leq
                                   c_2 \cdot (n + f(\bB(\bX)) )\;.
  \end{align*}

\medskip\noindent
  Let $\delta_1\triangleq\frac{c-3(c_2-c_1)}{c_2}$. Note that
  $\delta_1=\Theta(c)$. Then,
  \begin{align*}
    \Pr{T_{b^2}(\bX) \geq \mu_{b^2} + c\cdot n} &\leq
\Pr{c_2 \cdot (n + f(\bB(\bX)))  \geq \mu_{b^2} + c\cdot n} \\
                                                &\leq
\Pr{c_2 \cdot (n + f(\bB(\bX)))  \geq c_1n+c_1\mu_f+cn}\\                       
&\leq \Pr{ f(\bB(\bX)) \geq \mu_f + \frac{c-3(c_2-c_1)}{c_2}\cdot n}\\
&= \Pr{ f(\bB(\bX)) \geq \mu_f + \delta_1 n}\;.
\end{align*}
The second inequality is proved in a similar fashion.
\end{proof}

\section{Variants of Chernoff  Bounds}\label{app:chernoff}

\begin{theorem} Let $X_1, \ldots, X_n$ be independent binary random variables . Let  $X=\sum_{i=1}^{n} X_i$ and $\mu=\expectation{}{X}$. Then the following Chernoff bounds hold:

\medskip\noindent
1. For every $\delta>0$:
\begin{align}\label{eq:ch1}
\Pr{X \geq (1+\delta)\mu}&\leq \parentheses{\frac{e^{\delta}}{(1+\delta)^{(1+\delta)}}}^\mu \;.
\end{align}

\medskip\noindent
2. For $0<\delta\leq 1$, 
\begin{align}\label{eq:ch2}
\Pr{X \geq (1+\delta)\mu}&\leq e^{-\mu\delta^2/3}\;.
\end{align}

\medskip\noindent
3. For $\delta\geq 1$,
\begin{align}\label{eq:ch3}
\Pr{X \geq (1+\delta)\mu}&\leq e^{-\mu\delta/3}\;.
\end{align}

\medskip\noindent
4. For $\delta\geq e$,
\begin{align}\label{eq:ch4}
\Pr{X \geq (1+\delta)\mu}&\leq e^{-\mu \delta\ln(\delta)/2}\;.
\end{align}

\medskip\noindent
5. For $0<\delta\leq 1$,
\begin{align}\label{eq:ch5}
\Pr{X \geq (1-\delta)\mu}&\leq e^{-\mu \delta^2/3}\;.
\end{align}
\end{theorem}

\begin{proof}
The bounds in Eqs.~\ref{eq:ch1}, \ref{eq:ch2} and  \ref{eq:ch5} are proved
in~\cite{mitzenmacher2017probability}.  For $\delta>0$, define the
function $f(\delta)\triangleq (1+\delta)\ln(1+\delta)-\delta$ and
note that Eq.~\ref{eq:ch1} states that
$\Pr{X \geq (1+\delta)\mu}\leq \exp(-\mu\cdot f(\delta))$.  For
$\delta\geq 1$, we have that $f(\delta) \geq \delta/3$, which
proves Eq.~\ref{eq:ch3}.  For every $\delta>0$,
$f(\delta) \geq \delta\ln(\delta)/2$, which proves Eq.~\ref{eq:ch4}. 
The bound in  Eq.~\ref{eq:ch4} and its proof also appear in Chapter $10.1.1$ in~\cite{doerr2018tools}.

\end{proof}

We note that the bounds~\ref{eq:ch1}--\ref{eq:ch4} hold even when the parameter $\mu$ is an
upper bound on $\expectation{}{X}$. Moreover, they also hold when the
random variables $X_1, \ldots, X_n$ are negatively
associated~\cite[Thm. 3.1]{dubhashi2009concentration}. Indeed, in the
proofs of the Claims~\ref{claim:sum1} -- \ref{claim:sum4}, we apply
Eq.~\ref{eq:ch1} -- \ref{eq:ch4} to the random variable $\size{S_i}$, which
is a sum of negatively associated indicator variables.
\end{document}